\documentclass[final,5p,times,twocolumn,lefttitle]{elsarticle}
\usepackage{latexsym}
\usepackage{amsmath}
\usepackage[utf8x]{inputenc}
\usepackage[T1]{fontenc}
\usepackage{subdepth}
\usepackage{color}
\usepackage{soul}
\usepackage{csquotes}
\usepackage{amsthm}
\usepackage{amssymb}
\usepackage{tensor}
\usepackage{dsfont}
\usepackage{tikz-cd}

\newtheorem{Theorem}{Theorem}

\usepackage{color}
\usepackage{graphicx}
\usepackage{xcolor}
\definecolor{bluc}{cmyk}{1,0.9,0,0}
\definecolor{rossoCP3}{cmyk}{0,.88,.77,.40}
\definecolor{rosso}{cmyk}{0,1,1,0.4}
\definecolor{rossos}{cmyk}{0,1,1,0.55}
\definecolor{rossoc}{cmyk}{0,1,1,0.2}
\definecolor{verdes}{cmyk}{0.92,0,0.59,0.4}
\usepackage[normalem]{ulem}
\usepackage{hyperref}
\hypersetup{colorlinks, bookmarksopen, bookmarksnumbered, citecolor=verdes, linkcolor=bluc, pdfstartview=FitH, urlcolor=rossos}

\begin{document}

\begin{frontmatter}
\title{S-Duality and Categorical Gauge Theory}
\author[First]{Javier Chagoya}
\ead{javier.chagoya@fisica.uaz.edu.mx}
\author[Second]{A.D. L\'opez-Hern\'andez}
\ead{ad.lopez.hernandez@ugto.mx}
\author[Second,Third]{M. Sabido}
\ead{msabido@fisica.ugto.mx}
\affiliation[First]{organization={Unidad Acad\'emica de F\'isica, Universidad Aut\'onoma de Zacatecas},
addressline={Calzada Solidaridad esquina con Paseo a la Bufa S/N}, city={Zacatecas},
            postcode={98060}, 
            state={Zacatecas},
            country={ M\'exico}}
\affiliation[Second]{organization={Departamento de F\'isica de la Universidad de Guanajuato},
addressline={A.P. E-143}, city={Le\'on},    
            postcode={37150},
            state={Guanajuato},
            country={M\'exico}}   
\affiliation[Third]{organization={Department of Theoretical Physics, University of the Basque Country UPV/EHU},
addressline={P.O BOX 644}, city={},
            postcode={48080}, 
            state={Bilbao},
            country={Spain}} 
\begin{abstract}
{In this work, we revisit abelian S-duality in the context of 
higher gauge theory. By using a specific crossed module
a  set of transformations arise, which are known as the {\it ``thin''} and {\it ``fat''} transformations. The {\it ``fat''}  transformations are the ones  introduced by hand to construct the S-dual theory.  By utilizing crossed modules and higher gauge theory, we arrive at transformations referred to by some as “third-type transformations”, giving these transformations a geometric interpretation,
as transformations of a local 2-connection in a 2-principal bundle. This approach opens a new  perspective to understand S-duality  in non-abelian gauge theories through higher gauge theory.} 
\end{abstract}
\end{frontmatter}

\section{Introduction}
One of the most developed tools to understand non perturbative effects in field theory is S-duality. It allows to compute in the strong coupling limit, it relates different models of string theory, and is one of the fundamental dualities in the context of M-theory. S-duality was originally realized in the context of supersymmetric theories, i.e $\mathcal{N}=4$ Yang-Mills theory, $\mathcal{N}=2$ SUSY QCD, and in  perturbative superstrings theory. This idea has also been explored in non-supersymmetric gauge theories, where the best example of S-duality in non-supersymmetric gauge theory is the abelian case \cite{witten1995s,Verlinde:1995mz,Lozano:1995aq}, where one considers a $U(1)$ gauge theory with a $\theta$ vacuum coupled. Moreover, the partition function $Z(\tau)$ transforms as a modular form under a finite index subgroup of SL(2,{\bf Z}), where $\tau$ is constructed with the electromagnetic coupling constant $g$ and usual $\theta$ angle. It is worth mentioning, that for the non-abelian cases \cite{Ganor:1995em},  
there is no dual theory in the same sense as for abelian theories, nonetheless, one usually follows the same procedure as in the abelian case. First we construct an intermediate Lagrangian from which one recuperates the original Lagrangian and its dual. In the same manner, S-duality in gravity has been explored \cite{Garcia-Compean:1997qar,Garcia-Compean:1999fdl,Garcia-Compean:2000uho}, by using a Yang-Mills type formulation for gravity.\\
S-duality incorporates a transformation of auxiliary degrees of freedom in a way that cannot be fully captured by ordinary principal bundles and connections. Higher gauge theory provides tools to address this, such as strict $2$-principal bundles~\cite{lopez2025categorical, baez2005higher, baez2011invitation, schreiber2008connections, schreiber2011smooth}, which make it possible to describe these structures in a consistent way. In general terms this generalization introduces a 2-form connection, in addition to the 1-form connection of conventional gauge theory. While the 1-form connection gives to curves holonomies in a gauge group, the 2-form connection is used to give to surfaces a new type of surface holonomy, represented by elements of another group. The group structure in this theory is called strict 2-groups (a categorical generalization of a group), which are equivalent to the so called crossed modules. This generalization faces a problem. Ensuring the consistency of the corresponding parallel transport requires that the so-called 2-form fake curvature must vanish \cite{baez2005higher,schreiber2011smooth}. However, some recent works \cite{Kim_2020, krist2022explicit} have addressed this problem by including what they call an adjustment, which is reflected 
in the 2-group structure and the connection.
Some physical applications of higher gauge theories 
explore its use in reformulating gravity 
and gauge theories. For instance, a particular categorization of $BF$ theory has been shown to 
allow the coupling of topological matter fields to 2+1 and 3+1 general relativity~\cite{Girelli:2007tt}. A different approach obtains dynamical degrees of freedom after imposing a simplicity constraint. Specifically, it has been shown that a classical action describing the Standard Model coupled to Einstein Gravity can be obtained from a constrained $BF$ theory based on a $3$-group~\cite{Radenkovic:2019qme}. Furthermore, focusing on deriving a theory of modified gravity from Categorified $BF$, a theory whose metric field equations are the same as those of unimodular gravity was presented in~\cite{lopez2025categorical}.\\
The purpose of this work is to study S-duality for non supersymmetric theories in the context of categorical gauge theory. The presence of a second connection can enhance the gauge transformations and can allow the realization of S-duality. The end result is that for abelian categorical gauge theory, the gauge theory associated to the first connection is S-dual to the gauge theory of the second connection. 
The manuscript is organized as follows. In Sec.~\ref{categorical_s}, 
we 
categorize the principal bundle to introduce the notion
of 2-principal bundles and their associated 2-form connections to construct the categorical abelian gauge theory (the formal details are in the appendix~\ref{appendix}). After constructing a Lagrangian invariant under the categorical gauge transformation, we make the connection to abelian S-duality. Lastly, Sec.~\ref{conclusions} is devoted to  final remarks.

\section{Categorical gauge theory and S-duality}\label{categorical_s}

The generalization of gauge theory, known as higher gauge theory~\cite{lopez2025categorical, baez2005higher, baez2011invitation, schreiber2008connections, schreiber2011smooth}, is in general terms, an extension of conventional gauge theory that introduces a 2-form connection in addition to the 1-form connection.  
The group structure in this theory is a strict 2-group (a categorical generalization of a group), which is equivalent to a crossed module, a less abstract formulation. A crossed module denoted by $(G,H,\rho ,\triangleright)$ consists of two Lie groups $G$ and $H$ with a group homomorphism $\rho: H \rightarrow G$ and an left action of $G$ on $H$. 
\footnote{The homomorphism and the action satisfy the Peiffer identity $\rho(h) \triangleright f=hfh^{-1}$ and the equivariant condition $\rho(g\triangleright h)=g\rho (h)g^{-1}$ for all $g\in G$ and $h,f\in H$.}

On the other hand, an important concept arising from crossed modules is a differential crossed module, which in a $2$-Lie group plays an analogous role as  that of a Lie algebra in an ordinary Lie group. A differential crossed module $(\mathfrak{g},\mathfrak{h},\rho _{\ast}, \triangleright ^{\prime})$ over a crossed module $(G,H,\rho ,\triangleright)$ consists of the Lie algebras of $G$ and $H$ respectively, along with a Lie algebra homomorphism $\rho _{\ast}:\mathfrak{h}\rightarrow \mathfrak{g}$ induced by the differential of $\rho$, and a left action of $\mathfrak{g}$ on $\mathfrak{h}$ induced by the action of $G$ on $H$.\\
The crossed module of interest in this paper is $({\rm SO}(2),U(1),\rho,\triangleright _{\rho})$ where
\begin{equation}
  \rho(e^{\frac{i}{2}\alpha})=
  \begin{pmatrix}
        \cos(\alpha) & -\sin(\alpha)\\
        \sin(\alpha) & \cos(\alpha)
  \end{pmatrix},
\end{equation}
and $g\triangleright _{\rho}h=h$ for all $g\in {\rm SO}(2)$. The associated differential crossed module is $(\mathfrak{so}(2), \mathfrak{u}(1), \rho_{\ast}, \triangleright ^{\prime})$, where 
\begin{equation}
  \rho _{\ast}(\frac{i}{2}\alpha)=\alpha J,\;\;\;J=
  \begin{pmatrix}
        0 & -1\\
        1 & 0
  \end{pmatrix},
\end{equation}
is a isomorphism and $\triangleright ^{\prime}=0$. 
Additionally,
\begin{equation}
            \langle \frac{i}{2},\frac{i}{2} \rangle_{\mathfrak{u(1)}}=\langle J,J\ \rangle_{\mathfrak{so(2)}}=1,
\end{equation}
which are bilinear, symmetric, non-degenerate, invariant under $U(1)$ and ${\rm SO}(2)$, respectively, and invariant forms.\\        
From higher gauge theory, we obtain a local $2$-connection $(A,G)$ in a $2$-principal bundle, where $A$ is a $\mathfrak{g}$-valued $1$-form (local connection) and $G$ is a $\mathfrak{h}$-valued $2$-form, The corresponding curvatures can be expressed in the abelian case presented here as:
\begin{align}
        \mathcal{F}_{A,G}&:=F_{A}-\rho _{\ast}G \equiv dA-\rho _{\ast}(G),\nonumber \\
        \mathcal{G}_{A,G}&:=dG,
    \end{align}
referred to as fake curvature and $3$-curvature, respectively. There exist two types of transformations for the $2$-connection $(A,G)$, one corresponding the $G$ Lie group and the other involving the $H$ Lie group:

\begin{itemize}
    \item Thin gauge  transformations:
A smooth function ``gauge transformation'' $g:M\rightarrow {\rm SO}(2)$  induces the transformations.
        \begin{align}
            A & \longrightarrow A+g^{-1}dg,\nonumber \\
            G & \longrightarrow G,
        \end{align}
        which in turn cause the curvature $F$, the false curvature $\mathcal{F}$, and the $3$-form curvature $\mathcal{G}$ to transform as.
        \begin{align}
        \label{tdelgadas}
             F&\longrightarrow F,\nonumber \\ \mathcal{F}&\longrightarrow \mathcal{F},\nonumber \\
             \mathcal{G}&\longrightarrow \mathcal{G}.
        \end{align}
   \item Fat gauge transformations:
Given a $\mathfrak{so}(2)$-valued $1$-form $B$, the transformations are:
        \begin{align}
        \label{fattransf}
            A&\longrightarrow A+B,\nonumber \\
            G& \longrightarrow G +d(\rho_{\ast}^{-1}B),
        \end{align}
        under which the curvature $F$, the false curvature $\mathcal{F}$, and the $3$-form curvature $\mathcal{G}$ transform as.
        \begin{align}
            F & \longrightarrow F+dB,\nonumber \\ \mathcal{F} & \longrightarrow \mathcal{F}, \nonumber \\
            \mathcal{G} & \longrightarrow \mathcal{G}.
        \end{align}
\end{itemize}
Once the symmetries have been  established, now we construct and invariant action under thin and fat gauge transformations. First because $\rho _{\ast}$ is an isomorphism and act only in the Lie algebra part, we can write without loss of generality $\rho{\ast}(G)$ as $G$ and $\rho_{\ast}^{-1}B$ as $B$. Thus we can define the field $\mathcal{F}=F-G$, and write the usual kinetic term plus the topological term for the 2-form $\mathcal{F}$ as the invariant action
\begin{equation}
\label{eq:accioncategorica}
        S=\int_M \frac{1}{g^2}\mathcal{F} \wedge \star \mathcal{F}+\frac{i\theta}{8\pi ^{2}}\mathcal{F} \wedge  \mathcal{F}.
\end{equation}

Moreover, one can write the action using the (anti)self-dual fields. Using the $\star$ is the Hodge operator $F^\pm$,  is defined  as  $2F^\pm=F\pm \star F$  and similar expression for $G^\pm$. After defining the complex coupling is defined as
    \begin{equation}
        \tau =\frac{\theta}{2\pi}+\frac{4\pi i}{g^{2}},
    \end{equation}
the action in Eq.\eqref{eq:accioncategorica} is given by
\begin{equation}\label{FF}
    \begin{split}
        S&=\frac{i}{4\pi}\int_M \left(\bar{\tau} \mathcal{F}^{+}\wedge   \mathcal{F}^{+}+\tau \mathcal{F}^{-}\wedge  \mathcal{F}^{-}\right)\\
        &=\frac{i}{8 \pi}\int_M d^{4}x\sqrt{g}\left(\bar{\tau} \mathcal{F}^{+}_{mn} \mathcal{F}^{+mn}-\tau \mathcal{F}^{-}_{mn}\mathcal{F}^{-mn}\right).
    \end{split}
\end{equation}    
As already stated, the action in terms of $\mathcal{F}^+$ and $\mathcal{F}^-$ is invariant under thin and fat gauge transformations.\\

On the other hand, by introducing a dual $1$-form connection $V$ with a gauge transformation $V\to V- d\alpha$ and curvature\footnote{Using the curvature $W$, we can construct
 a generalized type $BF$ term which constrains $dG=0$. } $W=dV$, we can construct the action 
\begin{equation}\label{WG}
    \frac{1}{2\pi}\int_M W \wedge G,
\end{equation}
which is also invariant (modulo  $2\pi$) under the thin and fat gauge transformations\footnote{The curvature $W$ satisfies the Dirac quantization condition $\int_{S}W=2n\pi$,\quad $n=\pm 1,\pm2,\cdots$.}. Finally, adding Eq.(\ref{FF}) and Eq.(\ref{WG}), we construct the action
\begin{equation}\label{master}
        I=\frac{i}{4\pi}\int_M 2W \wedge G+\bar{\tau} \mathcal{F}^{+}\wedge  \mathcal{F}^{+}+\tau \mathcal{F}^{-}\wedge  \mathcal{F}^{-}.
\end{equation}
Next we need to show that that this action is equivalent to Maxwell theory. Integrating out $V$ in the functional integral we get that $dG=0$ and from the thin and fat gauge transformations, we  take the gauge $G=0$. Then Eq.(\ref{master}) reduces to
\begin{equation}\label{maxwell}
    \begin{split}
        I
        &=\frac{i}{4\pi}\int_\mathcal{M}\left(\bar{\tau} F^{+}\wedge 
        F^{+}+\tau F^{-}\wedge   F^{-}\right)\\
        &=\frac{i}{8 \pi}\int_\mathcal{M} d^{4}x\sqrt{g}\left(\bar{\tau}F^{+}_{mn} F^{+mn}-\tau F^{-}_{mn}F^{-mn}\right),
    \end{split}
\end{equation}
which is the Maxwell action, also the transformations reduce to the regular $U(1)$ gauge transformations.\\ 
Alternatively, because of the thin and fat  gauge transformations, we can take the gauge $A=0$ and write Eq.(\ref{master}) in terms of the (anti) self-dual curvatures
\begin{equation}
    \begin{split}
        I&=\frac{i}{4\pi}\int_M 2 \left(W^{+}\wedge  G^{+}+W^{-}\wedge  G^{-}\right)\\
    &+\left( \bar{\tau} G^{+}\wedge  G^{+}+\tau G^{-}\wedge  G^{-}
    \right).
    \end{split}
\end{equation}
We  integrate out the fields $G^+$ and $G^-$ in the path integral, this leaves a Lagrangian that only contains $W$ (for details see ~\cite{witten1995s}),   
    \begin{align}\label{eq:dualw}
        I&=\frac{i}{4\pi}\int_\mathcal{M} \left(-\frac{1}{\bar{\tau}} {W}^{+}\wedge  {W}^{+}-\frac{1}{{\tau}} {W}^{-}\wedge  {W}^{-}\right)\\
        &=\frac{i}{8\pi}\int_M d^{4}x\sqrt{g}\left[\left(\frac{-1}{\bar{\tau}}\right )  W_{mn}^{+}W^{+mn} \right.\nonumber \\ &\left. -\left(\frac{-1}{\tau }\right)W_{mn}^{-}W^{-mn}
        \right],\nonumber
    \end{align}
this is the Maxwell action, with the regular gauge symmetry but with coupling constant $\tau\rightarrow -1/\tau$. Therefore we have constructed the Maxwell Lagrangian with and inverted coupling constant.\\
One might wonder why not simply use a crossed module of the form $(G,G,id_{G},\triangleright)$ which is denoted $INN(G)$~\cite{roberts2007inner}. However, this crossed module turns out to be equivalent to the trivial one, rendering the resulting theory purely gauge.\\ 
If we tried to construct the 1-form $V$ through the 2-connection $G$, we encounter 
the following problems. There are two different ways to construct an $1$-form from $G$,
 $V = \star dG$ (pseudovector) or $V = \star d \star G$. The first is invariant under thin and 
fat gauge transformations, while $\star d \star G$ is not. So if we choose $V = \star dG$
and impose the gauge transformation $V \rightarrow V - d\alpha$, we need the 2-connection 
$G$ to transform as $G \rightarrow G + \gamma$, such that $d\gamma = * d\alpha$,
this is because $dG$ needs to be the curvature of the 2-connection. Now, this imposes the condition
$d\star d\alpha = 0$, which is the transformation condition between connections in the Lorenz gauge.
Also, the invariant term $dG\wedge \star dG$  can be rewritten as $V\wedge \star V$ and is a term that transforms 
as $V\wedge \star V\rightarrow V\wedge \star V+d(\alpha \wedge (2dG+\star d\alpha))$. While the terms involving $\mathcal{F}$ in Eq.~\eqref{eq:accioncategorica}, transform as
\begin{equation}
\begin{split}
        &\frac{1}{g^2}\mathcal{F} \wedge \star \mathcal{F}+\frac{i\theta}{8\pi ^{2}} \mathcal{F} \wedge  \mathcal{F}\rightarrow 
        \frac{1}{g^2}\mathcal{F} \wedge \star \mathcal{F}+\frac{i\theta}{8\pi ^{2}}\mathcal{F} \wedge  \mathcal{F}\\
        &+(\gamma-2\mathcal{F})\wedge (\frac{\star \gamma}{g^{2}}+\kappa\gamma),\nonumber
\end{split}
\end{equation}
where $\kappa=\frac{i\theta}{8\pi ^{2}}$.\\ 
For the action to remain invariant there are two possibilities. We can have $d(\gamma-2\mathcal{F})=0$ and $d(\star \gamma+g^2\kappa\gamma)=0$, which leads to $W=dV=d\star dG=\frac{1}{2}d^{2}\alpha=0$, meaning that $V$ is a pure gauge. {Alternatively, we can have  $(\star \gamma+g^2\kappa\gamma)=0$, this can be achieved with $(\theta ^{2}g^{4}+8^2\pi^4)\gamma=0$} but this implies that  $\theta$ is imaginary. This imposes that $\gamma =0$ and consequently $d\alpha =0$ and $V$ does not admit a gauge transformation, this should be expected as the theory  has mass-like term for $V$.\\
{We can arrive to the same conclusion if we consider the field equations. By varying respect to $G$ and $A$,  
after fixing the gauge with $A=0$ and using the operator $\star d\star$ we obtain
\begin{equation}
    \star d\star W=\frac{g^{\prime \;2}}{g^{2}}\left[1+\frac{\theta^2g^{4}}{64\pi ^{4}}\right]V,
\end{equation}
which is a Proca-type equation in Euclidean space.
It may be surprising that the term $\theta$ appears in the classical equation, but this can be traced to the fact that 
the $\theta$-term constructed through $\mathcal{F} = dA - G$ is not a topological term unless $dG = 0$. On the other hand, the lack of an equation of the form $dG=0$ makes it difficult to choose an appropriate gauge for $G$ and to find the analog of a dual theory.}

\section{Final remarks}\label{conclusions}
In this paper we studied  categorical gauge theory in connection to S-duality.
The use of categorical gauge theory, enhances the gauge transformations, as not only the usual {\it``thin"}  gauge transformations are present but also a new set of {\it``fat"} gauge transformations. These are the generalized transformation introduced by hand in \cite{witten1995s} to construct abelian S-duality. By utilizing crossed modules and higher gauge theory, we arrive at transformations referred to by some as ``third-type transformations'', giving  these transformations a geometric interpretation, as transformations of a local 2-connection in a 2-principal bundle. Then after introducing a dual connection, one can construct a dual theory with and inverted coupling constant. Moreover, following the same steps as in \cite{witten1995s}, the partition function is modular invariant and therefore concluding the construction of the S-dual theory. {Additionally, this ``technique'' offers the possibility of obtaining results for the non-abelian case by using a crossed module $({\rm Spin}(n),{\rm SO}(n),\rho,\triangleright)$.
For instance, for $n = 3$ we have the crossed module $({\rm SU}(2),{\rm SO}(3),\rho,\triangleright)$.\\}
In summary, S-duality seems to be a natural property in abelian categorical gauge theory and might be the appropriate construction to study S-duality in non-supersymmetric abelian gauge theory. Considering that S-duality in the non abelian case is not well understood (except in some some supersymmetric cases) it is interesting to pursue a construction based on categorical gauge theory. Results in this direction are under research and will be presented elsewhere.

\section{Appendix}\label{appendix}
 
\noindent Just as a local connection can be viewed as a functor, a local 2-connection can be seen as a strict 2-functor (see~\cite{baez2011invitation, schreiber2011smooth, lopez2025categorical} for details)
\begin{equation}
    hol:\mathcal{P}_{2}(U)\rightarrow \mathcal{G}
\end{equation}
for some Lie strict 2-group $\mathcal{G}$. This 2-functor is equivalent to a tuple $(A,\beta)$ where $A$ is a $\mathfrak{g}$-valued 1-form (local connection) and $\beta$ is a $\mathfrak{h}$-valued 2-form. With this 2-connection, the corresponding curvature is
\begin{align}
        \mathcal{F}_{A,\beta}&=F_{A}-\partial _{\ast}\beta \equiv dA+A\wedge A-\partial _{\ast}\beta,\nonumber \\
        \mathcal{G}_{A,\beta}&=d\beta +A\wedge ^{\triangleright ^{\prime}}\beta,
    \end{align}
referred to as fake curvature and $3$-curvature, respectively.\\
\noindent In this generalization, gauge transformations are described by what is known as a natural pseudotransformation between the functors of holonomy:
\begin{Theorem}
    Let $hol ^{\prime}, hol:\mathcal{P}_{2}(U)\rightarrow \mathcal{G}$ be smooth $2$-functors with associated $1$-forms $A^{\prime},A\in \Omega^{1}(U,\mathfrak{g})$ and $2$-forms $\beta ^{\prime},\beta \in \Omega^{2}(U,\mathfrak{h})$ respectively. The smooth function $g:U\rightarrow G$ and the $1$-form $\eta \in\Omega^{1}(U,\mathfrak{h})$ extracted from a smooth pseudonatural transformation $\rho:hol^{\prime}\rightarrow hol$ satisfy the relations
    \begin{equation}
        A+\partial _{\ast}(\eta)=gA^{\prime}g^{-1}-(dg)g^{-1}
    \end{equation}
    \begin{equation}
        \beta +A\wedge ^{\triangleright ^{\prime}}\eta+d\eta+\eta \wedge\eta=g\triangleright ^{\prime \prime}\beta ^{\prime}
    \end{equation}
\end{Theorem}
\begin{proof}
    See  Schreiber and Waldorf~\cite{schreiber2011smooth}
\end{proof}
\noindent If we choose a smooth function $g:U\rightarrow G$ and a $1$-form $\eta \in\Omega^{1}(U,\mathfrak{h})$ that satisfies the relations of the previous theorem, then as shown in \cite{schreiber2011smooth}, a pseudonatural transformation $\rho:hol\rightarrow hol^{\prime}$ can be defined, with $g$ and $\eta$ as its extracted data. Thus we can define two type of transformations, one resembling a usual gauge transformation and another one that is proper of the generalization.

\begin{itemize}
    \item Thin gauge transformations:
A smooth function ``gauge transformation'' $g:M\rightarrow G$ with $\eta =0$ induces the transformations
        \begin{align}
            A & \longrightarrow g^{-1}Ag+g^{-1}dg,\nonumber \\
            \beta & \longrightarrow g^{-1}\triangleright ^{\prime \prime}\beta,
        \end{align}
        which in turn cause the curvature $F$, the fake curvature $\mathcal{F}$ and the $3$-form curvature $\mathcal{G}$ to transform as
        \begin{align}
             F&\longrightarrow g^{-1}Fg,\nonumber \\ \mathcal{F}&\longrightarrow g^{-1}\mathcal{F}g,\nonumber \\
             \mathcal{G}&\longrightarrow g^{-1}\triangleright ^{\prime \prime}\mathcal{G}.
        \end{align}
    \item Fat gauge transformations:
Given a $\mathfrak{h}$-valued 1-form $\eta$, the transformations with $g:U\rightarrow G$ trivial,
are
        \begin{align}
            A&\longrightarrow A+\partial _{\ast}(\eta),\nonumber \\
            \beta& \longrightarrow \beta +d\eta+A\wedge ^{\triangleright ^{\prime}}\eta +\eta \wedge \eta,
        \end{align}
        under which the curvature $F$, the fake curvature $\mathcal{F}$ and the $3$-form curvature $\mathcal{G}$ transform as
        \begin{align}
        \label{tgruesas}
            F & \longrightarrow F+\partial _{\ast}(d\eta+A\wedge ^{\triangleright ^{\prime}}\eta +\eta \wedge \eta),\nonumber \\ \mathcal{F} & \longrightarrow \mathcal{F}, \nonumber \\
            \mathcal{G} & \longrightarrow \mathcal{G}+\mathcal{F}\wedge ^{\triangleright ^{\prime}} \eta.
        \end{align}
        When $H$ is abelian the term $\eta \wedge \eta$ is zero.
        The gauge transformation group is given by all pairs $(g,\eta)$, and the group product is given by the semi-direct product
        \begin{equation*}
            (g,\eta)(g^{\prime},\eta ^{\prime})=(gg^{\prime},(g\triangleright ^{\prime}\eta ^{\prime})\eta).
        \end{equation*}
\end{itemize}

\section*{Acknowledgments}
This work is  supported by   SECIHTI DCF-320821. M. S. is supported by the SECIHTI program {\it``Est\'ancias sab\'aticas vinculadas a la consolidaci\'on de grupos de investigaci\'on''}.
\bibliographystyle{unsrt}
\bibliography{ref}
\end{document}